\documentclass[runningheads,a4paper]{llncs}
\usepackage{amssymb}
\usepackage{amsmath}
\usepackage{graphicx}
\usepackage{pgf}
\usepackage{paralist} 
\usepackage{graphicx}
\usepackage{boxedminipage}
\usepackage{boxedminipage}
\usepackage{xcolor}
\usepackage[linesnumbered,algosection]{algorithm2e}
\usepackage{framed}
\usepackage{todonotes}

\newtheorem{observation}{Observation}
\usepackage{url}

\newcommand{\wfvsfull}{\textsc{Weighted Feedback Vertex Set}}
\newcommand{\wfvs}{\textsc{Weighted-FVS}}

\newcommand{\diswfvs}{\textsc{Disjoint Weighted-FVS}}
\newcommand{\blockgraph}{block graph}
\newcommand{\fvssolfull}{feedback vertex set}
\newcommand{\fvssol}{fvs}
\newcommand{\wfvssolfull}{weighted feedback vertex set}
\newcommand{\wfvssol}{weighted-fvs}
\newcommand{\wfvsmath}{{weighted-fvs}}
\newcommand{\wmp}{\textsc{Weighted-Matroid Parity}}
\newcommand{\Bgvd}{\textsc{Block Graph Vertex Deletion}}
\newcommand{\bgvd}{\textsc{BGVD}}
\newcommand{\bvds}{block vertex deletion set}

\newcommand{\FVS}{\textsc{Feedback Vertex Set}}
\newcommand{\fvs}{\textsc{FVS}}
\newcommand{\compdeg}{\emph{component degree}}
\newcommand{\FPT}{\textsc{FPT}}

\newcommand{\OO}{\mathcal{O}}
\newcommand{\CC}{\mathcal{C}}
\newcommand{\BB}{\mathcal{B}}
\newcommand{\DD}{\mathcal{D}}
\newcommand{\II}{\mathcal{I}}
\newcommand{\RR}{\mathcal{R}}

\newcommand{\no}{\textsc{No}}
\newcommand{\yes}{\textsc{Yes}}
\newcommand{\NP}{\textsc{NP}}
\newcommand{\branchvector}[1]{{\color{red}{$(#1)$}}}

\newcommand{\defparproblem}[4]{
  \vspace{3mm}
\noindent\fbox{
  \begin{minipage}{.95\textwidth}
  \begin{tabular*}{\textwidth}{@{\extracolsep{\fill}}lr} \textsc{#1}  & {\bf{Parameter:}} #3 \\ \end{tabular*}
  {\bf{Input:}} #2  \\
  {\bf{Question:}} #4
  \end{minipage}
  }
  \vspace{2mm}
}

\newcommand{\defparproblemoutput}[4]{
  \vspace{1mm}
\noindent\fbox{
  \begin{minipage}{0.96\textwidth}
  \begin{tabular*}{\textwidth}{@{\extracolsep{\fill}}lr} #1  & {\bf{Parameter:}} #3
\\ \end{tabular*}
  {\bf{Input:}} #2  \\
  {\bf{Output:}} #4
  \end{minipage}
  }
  \vspace{1mm}
}

\newtheorem{redrule}{Reduction Rule}
\newtheorem{redrulebgvd}{Reduction Rule BGVD.}
\newtheorem{branchrule}{Branching Rule}

\title{A faster FPT Algorithm and a smaller Kernel for {\sc \Bgvd}}
\author{Akanksha Agrawal\inst{1} \and Sudeshna Kolay\inst{2} \and Daniel Lokshtanov\inst{1} \and Saket Saurabh\inst{1,2}}
\institute{University of Bergen, Norway \\  \texttt{\{akanksha.agrawal|daniello\}l@uib.no} \and Institute of Mathematical Sciences, Chennai, India \\ \texttt{\{skolay|saket\}@imsc.res.in}
}

\begin{document}
\maketitle
\begin{abstract}
A graph $G$ is called a \emph{block graph} if each maximal $2$-connected component of $G$ is a clique. In this paper we study the \Bgvd\ from the perspective of fixed parameter tractable (FPT) and kernelization algorithm. In particular, an input to \Bgvd\ consists of a graph $G$ and a positive integer $k$ and  the objective to check whether there exists a subset $S \subseteq V(G)$ of size at most $k$ such that the graph induced on $V(G)\setminus S$ is a block graph. In this paper we give an FPT algorithm with running time $4^kn^{\OO(1)}$ and a polynomial kernel of size $\OO(k^4)$ for \Bgvd. The running time of our FPT algorithm improves over the  previous best algorithm for the problem that ran in time $10^kn^{\OO(1)}$ and the size of our kernel reduces over the previously known kernel of size $\OO(k^9)$. Our results are based on a novel connection between \Bgvd\ and the classical {\sc Feedback Vertex Set} problem in graphs without induced $C_4$ and $K_4-e$. To achieve our results we also obtain an algorithm for {\sc Weighted Feedback Vertex Set} running in time 
$3.618^kn^{\OO(1)}$ and improving over the running time of previously known algorithm with running time $5^kn^{\OO(1)}$. 

\end{abstract}
\section{Introduction}
Deleting minimum number of vertices from a graph such that the resulting graph belongs to a family $\mathcal{F}$ of graphs, is a measure on how close the graph is to the graphs in the family $\mathcal{F}$. In the problem of vertex deletion, we ask whether we can delete at most $k$ vertices from the input graph $G$ such that the resulting graph belongs to the family $\mathcal{F}$. Lewis and Yannakakis~\cite{Lewis} showed that for any non-trivial and hereditary graph property $\Pi$ on induced subgraphs, the vertex deletion problem is NP-complete. Thus these problems have been subjected to intensive study in algorithmic paradigms that are meant for coping with NP-completeness~\cite{FominLMS12,Fujito98,LundY94,MOR13}. These paradigms among others 
include applying restrictions on inputs, approximation algorithms and parameterized complexity. The focus of this paper is to study one such problem from the view point of parameterized algorithms. 

Given a family $\cal F$, a typical parameterized vertex deletion problem gets as an input an undirected graph 
$G$ and a positive integer $k$ and the goal is to test whether there exists a vertex subset $S\subseteq V(G)$ of size at most $k$ such that $G\setminus S \in \cal F$. In the parameterized complexity paradigm the main objective is to design an algorithm for the vertex deletion problem that runs in time $f(k)\cdot  n^{\OO(1)}$, where $n=\vert V(G) \vert$ and $f$ is an arbitrary computable function depending only on the parameter $k$. Such an algorithm is called an \FPT\ algorithm and such a running time is called \FPT\ running time. We also desire to design a polynomial time preprocessing algorithm that reduces the given instance to an equivalent one with size as small as possible. This is mathematically modelled by the notion of {\em kernelization}. A parameterized problem is said to admit a $h(k)$-{\it kernel} if there is a polynomial time algorithm (the degree of the polynomial is independent of $k$), called a {\em kernelization} algorithm, that reduces the input instance to an equivalent instance with size upper bounded by $h(k)$. If the function $h(k)$ is polynomial in $k$, then we say that the problem admits a polynomial kernel. The formal definitions are in Section~\ref{sec:prelims}. For more background, the reader may refer to the following monograph \cite{FlumGrohebook}.


A large part of research in paramaterized vertex deletion problems has centred around $\mathcal{F}$, that can be defined by finite excluded minor characterization and in particular when $\cal F$ has a planar graph~\cite{FominLMS12,LPRRSS13}. These  include problem of deleting $k$ vertices to get  a graph of treewidth $t$~\cite{FominLMS12,LPRRSS13}. One of the main reason to study these problems is that several NP-hard problems become polynomial time solvable on graphs of bounded treewidth. However, there are several other notions similar to treewidth where several NP-hard problems become polynomial time solvable; these include cliquewidth, rankwidth and linear rankwidth.  Recently, Kant\'{e} et al.~\cite{Kante15} studied the problem of deleting $k$ vertices to get a graph of linear rankwidth one and in an other study Kim and Kwon~\cite{bgvd-kim} studied  deleting $k$ vertices to get a {\em block graph}. 

A graph $G$ is known as a \emph{block graph} if every maximal $2$-connected component in $G$ is a clique. Equivalently, we can see a block graph as a graph obtained by replacing each edge in a forest by a clique. A \emph{chordal graph} is a graph which has no induced cycles of length at least four. An equivalent characterisation of a block graph is a chordal graph with no induced $K_4-e$~\cite{graph-classes,block-edward}. The class of block graphs is the intersection of the chordal and distance-hereditary graphs~\cite{block-edward}.

In this paper, we consider the problem which we call \Bgvd~(\bgvd). Here, as an input we are given a graph $G$ and an integer $k$, and the question is whether we can find a subset $S\subseteq V(G)$ of size at most $k$ such that $G\setminus S$ is a block graph. The \NP-hardness of the \bgvd~problem follows from~\cite{Lewis}. 

\defparproblem{ \Bgvd\  ({\bgvd})}{An undirected graph $G=(V,E)$, and a positive integer $k$}{$k$}{Is there a set $S\subseteq V$, of size at most $k$, such that $G\setminus S$ is a block graph?}
\vspace{10 pt}

\noindent
Kim and Kwon~\cite{bgvd-kim} gave an \FPT{} algorithm with running time $\OO^{\star}(10^k)$ and a kernel of size $\OO(k^9)$ for the \bgvd{} problem. In this paper we improve both these results via a novel connection to {\sc Feedback Vertex Set} problem. 
\medskip

\noindent 
{\bf Our Results and Methods.}  We start by giving the results we obtain in this article and then we explain how we obtain these results. Our two main results are:
\begin{theorem}\label{thm:bgvd-fpt}
 \bgvd{} has an \FPT{} algorithm running in time $\OO^{\star}(4^k)$.
\end{theorem}
\vspace{-0.3cm}
\begin{theorem}\label{thm:bgvd-kernel}
 \bgvd{} has a kernel of size $\OO(k^4)$.
\end{theorem}
Our two theorems improve both the results in~\cite{bgvd-kim}. That is, the running time of our FPT algorithm improves over the  previous best algorithm for the problem that ran in time $10^kn^{\OO(1)}$ and the size of our kernel reduces over the previously known kernel of size $\OO(k^9)$.

Our results are based on  a connection between the \wfvs{} and \bgvd{} problems. In particular we show that if a given input graph does not have induced four cycles or diamonds ($K_4-e$) then we can construct an auxiliary bipartite graph and solve \wfvs{} on it. This results in a faster \FPT{} algorithm for \bgvd. In the algorithm that we give for the \bgvd{} problem, as a sub-routine we use the algorithm for the \wfvs{} problem. For obtaining a better polynomial kernel for \bgvd{}, most of our Reduction Rules are same as those used in \cite{bgvd-kim}. On the way to our result we also design a factor four approximation algorithm for \bgvd.


Finally, we talk about \wfvs. For which, we also design a faster algorithm than known in the literature. 
The \FVS{} problem is one of the most well studied problems. Given an undirected graph $G =(V,E)$ and a positive integer $k$, the problem is to decide whether there is a set $S \subseteq V$ such that $G\setminus S$ is a forest. Thus, $S$ is a vertex subset that intersects with every cycle of $G$. 
In the parameterized complexity setting, \FVS{} parameterized by $k$, has an \FPT{} algorithm. The best known \FPT{} algorithm runs in time $\OO^{\star}(3.618^k)$ \cite{saurabh-book,fvs-matriod-marcin}. The problem also admits a kernel on $\OO(k^2)$ vertices \cite{Thomasse}. Another variant of \FVS{} that has been studied in parameterized complexity is \wfvsfull, where each vertex in the graph has some rational number as its weight.

\defparproblemoutput{{\wfvs}}{An undirected graph $G=(V,E)$, a weight function $w:V \rightarrow \mathbb{Q}$, and a positive integer $k$}{$k$}{The minimum weighted set $S \subseteq V$ of size at most $k$, such that $G\setminus S$ is a forest.}
\vspace{10 pt}

\wfvs{} is known to be in \FPT{} with an algorithm of running time $5^kn^{\OO(1)}$ \cite{Chen20081188}. We obtain a faster \FPT{} algorithm for \wfvs. This algorithm uses, as a subroutine, the algorithm for solving \wmp~\cite{Tong-weighted-matroid}. In fact, this algorithm is very similar to the algorithm for \FVS{} given in \cite{saurabh-book}. Thus, our final new result is the following theorem. 

\begin{theorem}\label{thm:wfvs-fpt}
 \wfvs{} has an \FPT{} algorithm running in time \newline $\OO^{\star}(3.618^k)$.
\end{theorem}
%
%
%
%
%
%

\section{Preliminaries}
We start with some basic definitions and terminology from graph theory and algorithms. We also establish some of the notation that will be used in this paper.

We will use the $\OO^\star$ notation to describe the running time of our algorithms. Given $f: \mathbb{N} \rightarrow \mathbb{N}$, we define $\OO^\star(f(n))$ to be $\OO(f(n) \cdot p(n))$, where $p(\cdot)$ is some polynomial function. That is, the $\OO^\star$ notation suppresses polynomial factors in the running-time expression.

\vspace*{2mm}
\noindent
\textbf{Graphs.} A graph is denoted by $G=(V,E)$, where $V$ and $E$ are the vertex and edge sets, respectively. We also denote the vertex set and edge set of $G$ by $V(G)$ and $E(G)$, respectively. All the graphs that we consider are finite graphs, possibly having loops and multi-edges. For any non-empty subset $W \subseteq V(G)$, the subgraph of $G$ induced by $W$ is denoted by $G[W]$; its vertex set is $W$ and its edge set consists of all those edges of $E(G)$ with both endpoints in $W$. For $W \subseteq V(G)$, by $G \setminus W$ we denote the graph obtained by deleting the vertices in $W$ and all edges which are incident to at least one vertex in $W$. 

For a graph $G$, we denote the degree of vertex $v$ in $G$ by $d_{G}(v)$. A vertex $v\in V(G)$ is called as a \emph{cut vertex} if the number of connected components in $G\setminus \{v\}$ is more than the number of connected components in $G$. For a vertex $v \in V(G)$, the neighborhood of $v$ in $G$ is the set $N_G(v)=\{u \lvert (v,u)\in E(G)\}$. We drop the subscript $G$ from $N_G(v)$, whenever the context is clear. Two vertices $u,v\in V(G)$ are called \emph{true-twins} in $G$ if $N(u)\setminus \{v\}=N(v)\setminus \{u\}$. For $A\subset V(G)$, an $A$-path in $G$ is a path with at least one edge, whose end vertices are in $A$ and all the internal vertices are from $V(G)\setminus A$.

Suppose that $u,v \in V(G)$, such that $u\neq v$ and neither $u$ nor $v$ has a self loop. By contracting an edge $(u,v)\in E(G)$ we mean the following operation. We create a new graph $G'$, where $V(G')=(V(G)\setminus \{u,v\})\cup \{uv^*\}$ and $E(G')=E(G[V(G)\setminus \{u,v\}]) \cup \{(uv^*,w)\lvert w \in (N(u)\cup N(v))\setminus \{u,v\}\}$. Note that there is bijection $f$ between $E(G)\setminus \{(u,v)\}$ and $E(G')$. The bijection $f$ is as follows. For $(x,y)\in E(G)$ $f((x,y))=(x,y)$ if both $x,y$ are not one of $\{u,v\}$. Otherwise, one of $x,y$ is same as $u,v$. Note that the edge $(x,y)$ is not same as the edge $(u,v)$, also since $u \neq v$ implies $x \neq y$. The only case left is when exactly one of $x$ or $y$ is same as one of $u,v$ say, $x=u$ then, $f((x,y))=(uv^*,y)$. Analogously, we can find $f((x,y))$ for the remaining cases. Slightly abusing the notation, for an edge $e\in E(G)\setminus \{(u,v)\}$ we will refer to $f(e)\in E(G')$ by $e$.

A weighted undirected graph is a graph $G=(V,E)$, with a weight function $w: V(G) \rightarrow \mathbb{Q}$. For a subset $X\subseteq V(G)$, $w(X)=\sum_{v\in X}w(v)$.

A \emph{feedback vertex set} is a subset $S\subseteq V(G)$ such that $G \setminus S$ is a forest.
A minimum weight \fvssolfull~of a weighted graph $G$ is a subset $X\subseteq V(G)$, such that $G \setminus X$ is a forest and $w(X)$ is minimum among all possible \wfvssol~in $G$. In a graph with vertex weights, an \fvs{} is called a \wfvssolfull{} (\wfvssol). Similarly, for a given positive integer $k$, a minimum \wfvssol{} of size $k$ is a subset $X \subseteq V(G)$ such that $\vert X \vert \leq k$, $G \setminus X$ is a forest and $w(X)$ is minimum among all possible \wfvssol{} in $G$ that are of size at most $k$. Given a graph $G$ 
and a vertex subset $S\subseteq V(G)$, we say that $S$ is a {\em \bvds} if $G-S$ is a block graph. 

A family $\mathcal{F}$ of sets over a finite universe $U$ is called a matroid if it satisfies the following properties.
\begin{itemize}
\item $\emptyset \in \mathcal{F}$,
\item if $A\in \mathcal{F}$ and $B\subseteq A$ then, $B \in \mathcal{F}$,
\item if $A,B \in \mathcal{F}$ and $\lvert A \rvert < \lvert B \rvert$ then, there exists $b \in B\setminus A$ such that $A\cup \{b\} \in \mathcal{F}$.
\end{itemize}
For a matroid $(U, \mathcal{F})$, the elements of $U$ are called the edges of the matroid and the sets $S \in \mathcal{F}$ are called as the independent sets. For an undirected graph $G$, a graphic matroid $\mathcal{M}_G$ is a matroid with $U=E(G)$ and a set $S\subseteq E(G)$ is an independent set in the matroid $\mathcal{M}_G$ if the graph $G'=(V(G),S)$ is acyclic. Note that $\mathcal{M}_G$ satisfies all the properties required for it to be a matroid.

A graph $G$ is called a \blockgraph{} if every maximal $2$-connected component of $G$ is a clique. A maximal $2$-connected subgraph in $G$ is called a \emph{block}. Another characterization of block graph is a graph which has no induced cycles of length more than $4$ and no induced $K_4-e$~\cite{graph-classes}. Here $K_4-e$ is a complete graph on $4$ vertices with one of the edges removed. For a graph $G$, let $V_c$ denote the set of cut vertices of $G$, and $\BB$ the set of its blocks. We then have a natural bipartite graph $F$ on $V_c \cup \BB$ formed by the edges $(v, B)$ if and only if $v \in V(B)$. Note that for a block graph $G$, $F$ is a forest~\cite{diestel-book}. The bipartite graph $F$ is called as block forest of $H$. We will arbitrarily root $F$ at some vertex $B\in V(F)$. 

 A leaf block of a \blockgraph{} $G$ is a maximal $2$-connected component with at most one cut vertex. For a maximal $2$-connected component $C$ in $G$ a vertex $v\in V(C)$ is called as an \emph{internal vertex} if $v$ is not a cut vertex in $G$. 

We refer the reader to~\cite{diestel-book} for details on standard graph theoretic notation and terminology we use in the paper.

\vspace*{2mm}
\noindent
\textbf{Parameterized Complexity.} A parameterized problem $\Pi$ is a subset of $\Gamma^{*}\times\mathbb{N}$, where $\Gamma$ is a finite alphabet. An instance of a parameterized problem is a tuple $(x,k)$, where $x$ is a classical problem instance, and $k$ is called the parameter. A central notion in parameterized complexity is {\em fixed-parameter tractability (FPT)} which means, for a given instance $(x,k)$, decidability in time $f(k)\cdot p(|x|)$, where $f$ is an arbitrary function of $k$ and $p$ is a polynomial in the input size.

\vspace*{2mm}
\noindent
\textbf{Kernelization.} A kernelization algorithm for a parameterized problem   $\Pi\subseteq \Gamma^{*}\times \mathbb{N}$ is an algorithm that, given $(x,k)\in \Gamma^{*}\times \mathbb{N} $, outputs, in time polynomial in $|x|+k$, a pair $(x',k')\in \Gamma^{*}\times \mathbb{N}$ such that (a) $(x,k)\in \Pi$ if and only if  $(x',k')\in \Pi$ and (b) $|x'|,k'\leq g(k)$, where $g$ is some computable function. The output instance $x'$ is called the kernel, and the function $g$ is referred to as the size of the kernel. If $g(k)=k^{\OO(1)}$ (resp. $g(k)=\OO(k)$) then we say that $\Pi$ admits a polynomial (resp. linear) kernel.\label{sec:prelims}
\section{Improved Algorithm for \wfvsfull}
In this section we give an improved algorithm for the \wfvs. We use the method of iterative compression together with branching and reduce \wfvs{} to the problem of \wmp, which can be solved in polynomial time. The algorithm we give is similar to the algorithm for \textsc{Feedback Vertex Set}~described in~\cite{saurabh-book,fvs-matriod-marcin}. We give an algorithm only for the disjoint variant of the problem. 

\begin{observation}[\cite{saurabh-book}] \label{obs:disj-actual}
The existence of an algorithm for the disjoint variant of \wfvs~with running time $c^{k}\cdot n^{\OO(1)}$, for a constant $c$, implies that \wfvs~can be solved in time $c^{k+1}\cdot n^{\OO(1)}$.
\end{observation}
In the \diswfvs, we are given an undirected graph $G=(V,E)$, a weight function $w: V(G) \rightarrow \mathbb{N}$, an integer $k$ and a \fvssolfull~$R\subseteq V(G)$ of size $k+1$. The objective is to find a set $X\subseteq V(G)\setminus R$, such that $X$ is a minimum weight \fvssolfull{} of size at most $k$ in $G$.

\subsection{Reduction Rules for {\sc  \diswfvs}}
Let $(G=(V,E),w,R,k)$ be an instance of \diswfvs~and let $F=G \setminus R$. We call a vertex $v\in V(F)$ a \emph{nice} vertex if $d_G(v) = 2$ and both its neighbors are in $R$. A vertex $v\in V(F)$ is called \emph{tent} if $d_G(v)=3$ and all its neighbors are in $R$.  We start with some simple reduction rules that preprocesses the graph. The Reduction Rules are applied in the order that they are described.

\begin{redrule}\label{redn:degree_less_than_one}
 Delete all vertices of degree at most one as they do not participate in any cycle.
\end{redrule}

\begin{redrule}\label{redn:must_be_in_soln}
 If there is a vertex $v\in V(F)$, such that $G[R\cup \{v\}]$ has a cycle, then include $v$ to the solution $X$, delete $v$ from $G$ and decrease $k$ by $1$.
\end{redrule}

\begin{redrule}\label{redn:edge_multiplicity}
 If there is an edge of multiplicity larger than $2$ in $E(G)$, then reduce its multiplicity to $2$.
\end{redrule}
The correctness of the above Reduction Rules do not depend on the weights and thus it is similar to the one for undirected version of \wfvs\ and can be found  in \cite{saurabh-book}.

\begin{redrule}\label{redn:leaf_in_F}
 Let $x \in V(F)$ be a leaf with the only neighbor as $y$ in $F$. Also, $x$ has at most $2$ neighbors in $R$. Subdivide the edge $(x,y)$ add the newly created vertex $x^*$ to $R$. We define the new weight function $w^*: V(G)\cup \{x^*\} \rightarrow \mathbb{Q}$, as follows: $w^*(x^*)=1$ and $w^*(v)=w(v)$, for $v\in V(G)$. Let $G^*$ be the newly created graph after subdivision of the edge $(x,y)$. Our reduced instance for \diswfvs{} is $(G^*,w^*,R\cup\{x^*\},k)$.
\end{redrule}

\begin{lemma} \label{sub-div-edge}
Reduction Rule~\ref{redn:leaf_in_F} is safe. 
\end{lemma}
\begin{proof}
Towards the proof we show that $G$ has a \wfvsmath~of size at most $k$ and weight at most $W$ if and only if $G^*$ has a \wfvsmath~of size at most $k$ and weight at most $W$.

In the forward direction, let $S \subseteq V(G) \setminus R$ be  a \wfvssol\ of  $G$ of size at most $k$ such that $w(S) \leq W$. We will show that indeed $S$ is a \wfvssol~in $G^*$ of size at most $k$ and $w^*(S) \leq W$. We first bound the weight of $S$ in $G^*$. Note that $w^*(v)=w(v)$, for $v \in V(G)$. Therefore $w^*(S)=w(S) \leq W$. We now show that $S$ is a \fvssolfull~in $G^*$. Suppose not, then there is a cycle $C$ in $G^*\setminus S$. If $C$ does not contain $x^*$. Then it also does not contain the edges $\{(x,x^*),(x^*,y)\}$. By definition, $C$ is also a cycle of $G \setminus S$, which is a contradiction. On the other hand, let $C$ contain the vertex $x^*$. By construction, $C$ must also contain the edges $\{(x,x^*),(x^*,y)\}$. However, this means that the cycle $C' = (C\setminus \{(x,x^*),(x^*,y)\}) \cup \{(x,y)\})$ is a cycle in $G \setminus S$, which is a contradiction. Therefore, $S \subseteq V(G) \setminus R$ is also a \wfvssol{} in $G^*$.


In the reverse direction, consider a \wfvssol{} $S\subseteq V(G^*) \setminus (R \cup \{x^*\})$ of size at most $k$ and $w^*(S) \leq W$. As $S$ is a \wfvssol~disjoint from $R \cup \{x^*\}$, $x^* \notin S$. Thus, $w(S)=w^*(S)$, . We now show that $S$ is a \fvssolfull~in $G$. Suppose not, then there is a cycle $C$ in $G\setminus S$. If $C$ does not contain the edge $(x,y)$, then by construction, $C$ is also a cycle of $G^* \setminus S$. This is a contradiction to the fact that $S$ was a \wfvssol{} of $G^*$. Otherwise, $C$ contains the edge $(x,y)$. But then, the cycle $C'=(C\setminus \{(x,y)\})\cup \{(x,x^*),(x^*,y)\}$ belongs to $G^* \setminus S$, which is a contradiction. Hence, $S$ must also be a \wfvssol{} for $G$.
\qed
\end{proof}

Now, we are ready to describe the main algorithm. To measure the running time of our algorithm for an instance $I=(G,w,R,k)$, we define the following measure.
\[\mu(I)=k+\rho(R)-(\eta+\tau)\]
Here, $\rho(R)$ is the number of connected components in $G[R]$ and $\eta,\tau$ are the number of nice vertices and tents in $F$, respectively.

Let $I=(G,w,R,k)$ be an instance where none of the  Reduction Rules~\ref{redn:degree_less_than_one}, \ref{redn:must_be_in_soln}, \ref{redn:edge_multiplicity} and \ref{redn:leaf_in_F} apply. It is clear that Reduction Rules~\ref{redn:degree_less_than_one}, \ref{redn:must_be_in_soln} and  \ref{redn:edge_multiplicity} do not increase the measure. We only need to worry about Reduction Rule~\ref{redn:leaf_in_F} which increases the number of vertices in $R$. 
The number of vertices in the resulting  $R$ increases and thus the number of connected component in $G^*[R]$ increases by one. However, we create either a \emph{nice} vertex or a \emph{tent} in $G^*$, therefore one of $\eta$ or $\tau$ increases by $1$. 
Hence, $\mu(I^*)= k + (\rho(R)+1) - (\eta+\tau+1) \leq \mu(I)$. Note that we do not increase the number of vertices in $F$, therefore we can apply the Reduction Rule~\ref{redn:leaf_in_F} at most $\lvert V(F) \rvert$ times.

In Lemma~\ref{mu-less-than-0}, we show that if $\mu <0$, then $(G,w,R,k)$ is a \no{} instance.  This will form a base case of our branching algorithm.

\begin{lemma} \label{mu-less-than-0}
For an instance $I=(G,w,R,k)$ of \diswfvs, if $\mu<0$, then $I$ is a \no{} instance.
\end{lemma}
\begin{proof}
Let us assume for contradiction that $I$  is a \yes{} instance and  $\mu<0$. Let $S$ be a \wfvssol~in $G$ of size at most $k$. Therefore, $F'=G\setminus S$ is a forest. Let $N\subseteq V(G)\setminus R$, $T\subseteq V(G)\setminus R$ be the set of nice vertices and tents in $V(G)\setminus R$, respectively. Since $F'$ is a forest we have that $G'=G[(R \cup N\cup T)\setminus S]$ is a forest. In $G'$, we contract each of the connected components in $R$ to a single vertex to obtain a forest $\tilde{F}$. Observe that $\tilde{F}$ has at most $\lvert V(\tilde{F})\rvert \leq \rho(R)+\lvert N \setminus S \rvert+\lvert T \setminus S \rvert$ vertices and thus can have at most $\rho(R)+\lvert N \setminus S \rvert+\lvert T  \setminus S \rvert-1$ many edges.
The vertices in $(N\cup T) \setminus S \subseteq V(G)\setminus R$ forms an independent set in $\tilde{F}$, since they are nice vertices or tents. 
The vertices in $N\setminus S$ and $T\setminus S$ have  degree $2$ and degree $3$ in $\tilde{F}$,  
respectively,  since their degree cannot drop while contracting the components of $G[R]$. This implies that, 
\begin{eqnarray*}
2 \lvert N \setminus S \rvert + 3\lvert T \setminus S \rvert & \leq |E(\tilde{F})| \leq & \rho(R)+ \lvert N \setminus S \rvert + \lvert T \setminus S \rvert -1.
\end{eqnarray*}
Therefore, $\lvert N \setminus S \rvert + \lvert T \setminus S \rvert < \rho(R)$. But $N \cap T = \emptyset$ and thus   
\begin{eqnarray}
\lvert N \rvert + \lvert T \rvert < \rho(R) + \lvert S \rvert \leq \rho(R)+k.
\label{eqn:measure}
\end{eqnarray}
However, by our assumption,  $\mu(I) = \rho(R) + k - (\lvert N \rvert + \lvert T \rvert)<0$ and thus $\lvert N \rvert + \lvert T \rvert > \rho(R) + k$. This, contradicts the inequality given in Equation~\ref{eqn:measure} contradicting our assumption that $I$ is a \yes{} instance. This completes the proof. 
\qed
\end{proof}

\subsection{Algorithm for {\sc  \diswfvs}}
In this section we conjure all that we have developed and give the description of the algorithm for 
\diswfvs, prove its correctness and analyze its running time assuming a polynomial time procedure that we explain in the next subsection. 

\medskip

\noindent \textbf{Description of the Algorithm.} Let $I=(G,w,R,k)$ be an instance of \diswfvs. If $G[R]$ is not a forest, then return that $(G,w,R,k)$ is a \no{} instance. Hereafter, we will assume that $G[R]$ is a forest. First, the algorithm exhaustively applies the Reduction Rules \ref{redn:degree_less_than_one} to \ref{redn:leaf_in_F}. If at any point $\mu(I)<0$, then we return that $(G,w,R,k)$ is a \no{} instance. For sake of clarity, we will denote the reduced instance by $(G,w,R,k)$. If all the vertices $v \in V(G\setminus R)$ are either nice vertices or tents then we solve the problem in polynomial time by using Theorem~\ref{weighted-matroid-parity-poly}. We defer the proof of Theorem~\ref{weighted-matroid-parity-poly} to the following subsection, where we solve the instance using \textsc{Weighted Matroid Parity}. Otherwise, we apply the following single Branching rule.


\begin{branchrule}\label{branch:no_good_not_nice}
{\rm 
 If there is a leaf vertex $v\in V(G)\setminus R$, which is neither a nice vertex nor a tent then, we branch as the following:
\begin{enumerate}[(i)]
\item $v$ belongs to the solution. In this branch we delete $v$ from $G$ and decrease $k$ by $1$. The resulting instance is $(G\setminus \{v\},w',R,k-1)$. Here, $w'$ is restriction of $w$ to $V(G)\setminus \{v\}$. 

\texttt{The measure $\mu$ decreases by at least  $1$}
\item $v$ does not belong to the solution. Note that $v$ is neither a nice vertex nor a tent and none of the Reduction rule applies. Therefore, $v$ has at least $3$ neighbors in $R$. We add the vertex $v$ to $R$. As a result, the number of components in $G[R]$ decreases by at least $2$. The resulting instance is $(G,w,R\cup \{v\},k)$. 

\texttt{The measure $\mu$ decreases by at least $2$.}
\end{enumerate}
}
\end{branchrule}

The worst case branching vector corresponding to the above branching rule is \branchvector{1,2}.

\begin{lemma}\label{correctness-wfvs}
The algorithm presented is correct.
\end{lemma}
\begin{proof}
Let $I=(G,w,R,k)$ be an input to the \diswfvs. We prove the correctness of the algorithm by induction on the measure $\mu=\mu(I)$. By Lemma~\ref{mu-less-than-0} when $\mu<0$, then we correctly conclude that $I$ is a \no{} instance. 

For induction hypothesis, let us assume that the algorithm correctly decides whether the input is a \yes{/}\no{} instance for $\mu=t$. We will prove it for $\mu=t+1$. If any of the Reduction Rules apply then we create an equivalent instance by the safeness of the Reduction rules. We either get an instance $I'$, where $\mu(I') < \mu(I)$ (the case when Reduction Rule~\ref{redn:must_be_in_soln} is applied) and then by induction hypothesis the algorithm correctly decides for the measure $\mu=t$. Otherwise, we have an instance with the same measure. If none of the reduction rules are applicable, then we have the following cases:
\begin{itemize}
\item Each $v \in V(G)\setminus R$ is either a nice vertex or a tent. In this case we solve the problem in polynomial time, and the correctness follows from the correctness of Theorem~\ref{weighted-matroid-parity-poly}.
\item There is a leaf $v\in V(G)\setminus R$, such that $v$ is neither a nice vertex nor a tent. In this case we apply the Branching Rule~\ref{branch:no_good_not_nice}. The Branching Rule is exhaustive. Moreover, at each branch the measure decreases at least by one. Hence, by induction hypothesis it follows that the algorithm correctly decides whether $I$ is a \yes{} instance or not.
\end{itemize}
This completes the proof of correctness. 
\qed
\end{proof}

Thus, we have an \FPT{} algorithm for \diswfvs.
\begin{lemma}\label{lem:diswfvs}
\diswfvs~can be solved in time $\OO^{\star}(2.618^k)$.
\end{lemma}
\begin{proof}
The correctness of the algorithm follows from Lemma~\ref{correctness-wfvs}. All of the Reduction rules~\ref{redn:degree_less_than_one} to \ref{redn:leaf_in_F} can be applied in polynomial time. Also, at each branch we spend a polynomial amount of time. For each of the recursive calls at a branch, the measure $\mu$ decreases at least by $1$. When $\mu<0$, then we are able to correctly conclude that the given input is a no instance by Lemma~\ref{mu-less-than-0}. The number of leaves and thus the size of the branching tree is upper bounded by the  solution to the following recurrence,  
\[T(\mu)\leq T(\mu-1)+T(\mu-2).\]
The above recurrence solves to $1.618^\mu$. Since at  the start of the algorithm $\mu \leq 2k$, we have that the number of leaves is upper bounded by $\OO(1.618^{2k})$. 
Therefore, \diswfvs~can be solved in time $\OO^{\star}(2.618^k)$.
\qed
\end{proof}

Using Lemma~\ref{lem:diswfvs} and Observation~\ref{obs:disj-actual}, we prove Theorem~\ref{thm:wfvs-fpt}.

\subsection{Algorithm for sub-cubic \diswfvs}
Let $(G,w,R,k)$ be an instance of \diswfvs~where each vertex in $V(G)\setminus R$ is either a nice vertex or a tent. An instance of the {\sc Matroid Parity} problem we create is same as that in~\cite{fvs-matriod-marcin}.
In fact, what we will use is the {\sc Weighted Matroid Parity} problem. 

The {\sc Weighted Matroid Parity} problem for the graphic matroid $\mathcal{M}_H$ of a graph $H$ is defined as follows. Let $H$ be a graph with even number of edges, i.e. $\lvert E(H)\rvert=2m$ and we have a partition of $E(H)$ into pairs, say $E(H)=\{e^1_1,e^1_2\} \cup \{e^2_1,e^2_2\} \cup \dots \cup \{e^m_1,e^m_2\}$. Furthermore, for each pair $\{e^i_1,e^i_2\}$, for $i \in \{1,2,\dots, m\}$ there is a positive weight $w_{\mbox{\tiny $\mathcal{M}$}}(\{e^i_1,e^i_2\})$. That is,  $w_{\mbox{\tiny $\mathcal{M}$}}$ is a weight function on pairs. We want to find a set $I \subseteq \{1,2,\dots, m\}$ of maximum weight such that $\cup_{i \in I} \{e^i_1,e^i_2\}$ is an independent set of $\mathcal{M}_H$. Equivalently, $\cup_{i \in I} \{e^i_1,e^i_2\}$ is acyclic in $H$. The {\sc Weighted Matroid Parity} problem is polynomial time solvable on graphical matroids (gammoids)~\cite{Tong-weighted-matroid}.

For each vertex $v \in V(G) \setminus R$, we arbitrarily label the edges incident to $v$. 
If $v$ is a nice vertex then we label it as $\{e^v_1,e^v_2\}$; otherwise if $v$ is a tent vertex then we label it as $\{e^v_0,e^v_1,e^v_2\}$. We let $w_{\mbox{\tiny $\mathcal{M}$}}(\{e^v_1,e^v_2\})=w(v)$. Note that, $F=E(G[R]) \cup \{e^v_0: v \in V(G)\setminus R\}$ is a forest. We contract all the edges in $G$ which are in $F$ to get a new graph $H$. In the process of contraction, we have not contracted any multiple edge or self loops. Also, $$E(H)=\bigcup_{v\in V(G)\setminus R}\{e^v_1,e^v_2\}.$$ 
The input to the \textsc{Weighted Matroid Parity} algorithm for graphical matroid is the graph $H$, the set of pairs $\{e^v_1,e^v_2\}$ with weight $w_{\mbox{\tiny $\mathcal{M}$}}(\{e^v_1,e^v_2\})$, for $v \in V(G)\setminus R$. In Lemma~\ref{correct-matroid-parity} we prove that finding a minimum weight \fvssolfull~$X\subseteq V(G)\setminus R$ in $(G,w,R,k)$ is equivalent to computing a maximum weight subset $I \subseteq \{\{e^v_1,e^v_2\}\lvert v \in V(G)\setminus R\}$, such that $\cup_{v \in I} \{e^v_1,e^v_2\}$ is an independent set in $\mathcal{M}_H$.

\begin{lemma} \label{correct-matroid-parity}
For a subset $I\subseteq V(G)\setminus R$, $\cup_{i \in I}\{e^i_1,e^i_2\}$ is an independent set in $\mathcal{M}_H$ of maximum weight if and only if $(V(G) \setminus R)\setminus I$ is a \fvssolfull~in $G$ of minimum weight.
\end{lemma}
\begin{proof}
Note that by the definition of $H$, $\cup_{i \in I}\{e^i_1,e^i_2\}$ is an independent set in $\mathcal{M}_H$ if and only if $F \cup (\cup_{i \in I}\{e^i_1,e^i_2\})$ is acyclic in $G$. As mentioned earlier, $F=E(G[R]) \cup \{e^v_0: v \in V(G)\setminus R\}$ is a forest. Therefore, if $\cup_{i \in I}\{e^i_1,e^i_2\}$ is an independent set in $\mathcal{M}_H$, then $G'=G\setminus ((V(G) \setminus R)\setminus I)$ is a forest. In other words, $(V(G) \setminus R)\setminus I$ is a \fvssolfull~in $G$. 

In the reverse direction, consider $I \subseteq V(G)\setminus R$ such that $(V(G) \setminus R)\setminus I$ is a \fvssolfull~in $G$. This implies that $G[I \cup R]$ is a forest. Define $F'=\bigcup_{i \in I}\{e^i_1,e^i_2\}$. Clearly, $F'\subseteq E(G[I \cup R])$. Suppose, $F'$ contains a cycle in $H$. This means that, upon uncontracting the edges of $F$, there is a cycle contained in $G[I\cup R]$, which is a contradiction. Therefore, $\cup_{i \in I}\{e^i_1,e^i_2\}$ is an independent set in $\mathcal{M}_H$.


Note that by definition of $w_{\mbox{\tiny $\mathcal{M}$}}$, $w_{\mbox{\tiny $\mathcal{M}$}}(I)=w(I)$. Therefore, $w((V(G) \setminus R)\setminus I)=w(V(G))-w(R)-w(I)=w(V(G))-w(R)-w_{\mbox{\tiny $\mathcal{M}$}}(I)$. This implies that whenever $w_{\mbox{\tiny $\mathcal{M}$}}(I)$ is maximized then $w((V(G) \setminus R)\setminus I)$ is minimized and vice-versa. This completes the proof.
\qed
\end{proof}

Lemma~\ref{correct-matroid-parity} immediately implies the following theorem.

\begin{theorem}\label{weighted-matroid-parity-poly}
Let $(G,w,R,k)$ be an instance of \diswfvs. If every vertex $v\in V(G)\setminus R$ is either a nice vertex or a tent then, \diswfvs~in $(G,w,R,k)$ can be solved in polynomial time.
\end{theorem}

\section{\FPT{} algorithm for \Bgvd}
In this section, we present an \FPT{} algorithm for the \bgvd~problem. First, we look at the special case, when the input graph does not have any small obstructions in the form of $D_4$'s and $C_4$'s. Here, $D_4 = K_4 -e$. We show that, in this case, \bgvd{} reduces to \wfvs. Later, we solve the general problem, using the algorithm of the special case.

\subsection{{\sc Restricted \bgvd}} \label{instance-weighted-fvs}

In this part, we solve the following special case of \bgvd{} in \FPT{} time.

\defparproblem{\sc Restricted \bgvd}{A connected undirected graph $G$, which is $\{D_4,C_4\}$-free, and a positive integer $k$.}{$k$}
{Does there exist a set $S$ such that $G\setminus S$ is a block graph?}

Let $G$ be the input graph. Let $\mathcal{C}$ be the set of maximal cliques in $G$. We start with the following simple observation about graphs without $C_4$ and $D_4$. 

\begin{lemma}\label{lem:maximal-cliques}
 Let $G$ be a graph that does not contain $C_4$ and $D_4$ as an induced subgraph then (a) any two maximal cliques intersect on at most one vertex and (b) 
 the number of maximal cliques in $G$ is at most $n^2$.
\end{lemma}

\begin{proof}
 Let $C_1$ and $C_2$ be two distinct maximal cliques in $\mathcal{C}$. Since $G$ is $D_4$-free, $V(C_1)\cap V(C_2)$ can have at most one vertex. Thus, each edge of $G$ belongs to exactly one maximal clique. This gives a bound of $n^2$ on the number of maximal cliques. \qed
\end{proof}

We construct an auxiliary weighted bipartite graph $\hat{G}$ in the following way: $\hat{G}$ is a bipartite graph with vertex set bipartition $V(G) \cup V_{\CC}$, where $V_{\CC}$ is the set where we add a vertex $v_{C}$ corresponding to each maximal clique $C \in \CC$. Note that there is a bijective correspondence between the vertices of $V_{\mathcal{C}}$ and the maximal cliques in $\mathcal{C}$. A vertex $v$ of a clique  $C$ is called {\em external} if it is part of at least two maximal cliques in $\cal C$. We add an edge between a vertex $v\in V(G)$ and a vertex $v_C \in V_{\mathcal{C}}$ in $E(\hat G)$ if and only if $v$ is an external vertex of the clique $C \in \mathcal{C}$. 


\begin{lemma}\label{lem:bgvd-wfvs}
Let $G$ be a graph without induced $C_4$ and $D_4$ and $S \subseteq V(G)$. Then $S$ is  \bvds\ of $G$ if and only if $\hat{G} \setminus S$ is acyclic.
\end{lemma}

\begin{proof}
First, let $S$ be a \bvds\ solution for $G$. Suppose that $\hat{G} \setminus S$ has a cycle $C$. Notice that $C$ cannot be a $C_4$, as this corresponds to two maximal cliques that share $2$ vertices. Thus, $C$ is an even cycle of length at least $6$. Suppose $C$ has length $6$. This corresponds to maximal cliques $C_1, C_2,C_3$ such that $u = C_1 \cap C_2$, $v=C_2 \cap C_3$ and $w = C_1\cap C_3$. Since $C_1,C_2,C_3$ are distinct maximal cliques, at least one of them must have a vertex other than $u,v$ or $w$. Without loss of generality, let $C_1$ have a vertex $x \notin\{u,v,w\}$. Then, the set $\{x,u,v,w\}$ forms a $D_4$ in $G$. However, this is not possible, as $G$ did not have a $D_4$ to start with. Hence, $C$ must be an even cycle of length at least $8$. However, this corresponds to a set of maximal cliques and external  vertices, such that the external vertices form an induced cycle of length at least four. This contradicts that $S$ was a \bvds\ for $G$. Thus, $\hat{G} \setminus S$ must be acyclic.
 
 On the other hand, let $\hat{G}\setminus S$ be acyclic. Suppose $G\setminus S$ has an induced cycle $C$, of length at least four. As $C$ is an induced cycle of length at least four, no two edges of $C$ can belong to the same maximal clique. For an edge $(u,v)$ of $C$, let $C_{(u,v)}$ be the maximal clique containing it. Also, let $c_{(u,v)}$ be the corresponding vertex in $\hat{G}$. We replace the edge $(u,v)$ in $C$ by two edges $(u,c_{(u,v)})$ and $(v,c_{(u,v)})$. In this way, We obtain a cycle $C'$ of $\hat{G}\setminus S$, which is a contradiction. Thus, $S$ must be a \bvds\ for $G$. \qed
 \end{proof}

If the input graph $G$ is without induced $C_4$ and $D_4$ then Lemma~\ref{lem:bgvd-wfvs} tells us that 
to find \bvds\ of $G$ of size at most $k$ one can check whether there is a feedback vertex set of size at most $k$ for $\hat{G}$ contained in $V(G)$. To enforce that we find feedback vertex set for $\hat{G}$ completely contained in $V(G)$  we solve an appropriate instance of \wfvs. In particular we give  the weight function $w: V(\hat G)\rightarrow \mathbb{Q}$ as follows. For $v\in V(G)$, $w(v)=1$ and for $v_{C}\in V_{\CC}$, $w(v_{C})= n^4$. Clearly, $V(G)$ is a feedback vertex set of $\hat{G}$ and thus the weight of a minimum sized  feedback vertex set of $\hat{G}$ is at most $n$. This implies that running an algorithm for 
\wfvs\ on an instance $(\hat{G},w,k)$ either returns a feedback vertex set contained inside $V(G)$ or returns that the given instance is a \no{} instance.

%
%

%

\begin{theorem}\label{thm:bgvd-wfvs}
 {\sc Restricted \bgvd} can be solved in $\OO^{\star}(3.618^k)$.
\end{theorem}

\begin{proof}
 We apply  \wfvs{} on an instance $(\hat{G},w,k)$. Let $S$ be the \wfvssol{} of size at most $k$ in $\hat{G}$ 
returned by \wfvs{} (of course if there exists one). By the discussion above we know that if \wfvs\ does not return that the given instance is a \no{} instance then $S\subseteq V(G)$. If it returns  that the given instance is a \no{} instance then we return the same. Else, assume that $S$ is non-empty. Now we check whether $w(S)$ is at most $k$ or not. Since every vertex in $V(G)$ has been assigned weight one we have that $w(S)=|S|$ and thus if $w(S)\leq k$ then we return $S$ 
 as \bvds\ of $G$.  In the case when  $w(S)> k$ we return that the given instance is a \no{} instance for  {\sc Restricted \bgvd}. Correctness of these steps are guaranteed by Lemma~\ref{lem:bgvd-wfvs}.
 The running time of the algorithm is dominated by the running time of \wfvs\ and thus it is $\OO^{\star}(3.618^k)$. This completes the proof. \qed
\end{proof}

\subsection{\Bgvd}
 
 We are now ready to describe an \FPT{} algorithm for \bgvd, and hence prove Theorem~\ref{thm:bgvd-fpt}. We design the algorithm for the general case with the help of the algorithm for {\sc Restricted \bgvd}.

 \begin{proof}[of Theorem~\ref{thm:bgvd-fpt}]
  Let $O$ be a $D_4$ or $C_4$ present in the input graph $G$. For any potential solution $S$, at least one of the vertices of $O$ must belong to $S$. Therefore, we branch on the choice of these vertices, and for every vertex $v\in O$, we recursively apply the algorithm to solve \bgvd\ instance $(G-\setminus \{v\},k-1)$. If one of these branches returns a solution $X$, then clearly $X\cup \{v\}$ is a \bvds\ of size at most $k$ for 
  $G$. Else, we return that the given instance is a \no{} instance.   
%
  On the other hand, if $G$ is $\{D_4,C_4\}$-free, then we do not make any further recursive calls. Instead, we run the algorithm for {\sc Restricted \bgvd} on $G$ and return the output of the algorithm. Thus, the running time of the algorithm is upper bounded by the following recurrence. 
  
  
  $T(n,k) = \left\{
\begin{array}{ll}
1 & \mbox{if } k=0 \mbox{ or } n=0 \\
3.168^k  & \mbox{if there is no } D_4,C_4 \\
4T(n-1,k-1) + n^{\OO(1)}& \mbox{otherwise}
\end{array}
\right.
  $
 
 Thus, the running time of this algorithm is upper bounded by $\OO^*(4^k)$. \qed
\end{proof}

\section{An Approximation Algorithm for \bgvd}
In this section, we present a simple approximation algorithm $\mathcal{A}_1$ for \bgvd. Given a graph $G$, we give a \bvds\ $S$ of size at most $4 \cdot {\sf OPT}$, where ${\sf OPT}$ is the size of a minimum 
sized \bvds\ for $G$. 

\begin{theorem}
\bgvd\ admits a factor four approximation algorithm. 
\end{theorem}
\begin{proof}
Let $G$ be the given instance of \bgvd\ and ${\sf OPT}$ be the size of a minimum  sized \bvds\ for $G$ and $S_{\sf OPT}$ be a minimum sized \bvds\ for $G$. 

Let $\cal S$ be a maximal family of $D_4$ and $C_4$ such that any two members of $\cal S$ are pairwise disjoint. One can easily construct such a family  $\cal S$ greedily in polynomial time.  Let $S_1$ be the set of vertices contained in any obstruction in $\cal S$. That is, $S_1=\bigcup_{O\in \cal S} O$.  
Since any \bvds\ must contain a vertex from each obstruction in  $\cal S$ and any two members of $\cal S$ are pairwise disjoint, we have that $|S_{\sf OPT} \cap S_1|  \geq |{\cal S}|$.

Let $G'=G-S_1$. Observe that $G'$ does not contain either $D_4$ or $C_4$ as an induced subgraph. 
Now we construct $\hat{G'}$, as described in Section~\ref{instance-weighted-fvs}. We apply the factor two approximation algorithm $\mathcal{A}$ given in \cite{2-approx-fvs-bafna} on the instance $(\hat{G'},w)$. This returns an \fvssol{} $S_2$ of $\hat{G'}$ such that $w(S_2)$ is at most twice the weight of a 
minimum weight feedback vertex set.  By our construction $S_2\subseteq V(G')$. Lemma~\ref{lem:bgvd-wfvs} implies 
that $S_2$ is a factor two approximation for \bgvd\ on $G'$. We return the set $S=S_1 \cup S_2$ as our solution. Since $S_{\sf OPT}-S_1$ is also an optimum solution for $G'$ we have that $|S_2|\leq  2 |S_{\sf OPT}-S_1|$.

It is evident that $S$ is \bvds\ of $G$. To conclude the proof of the theorem we will show that 
$|S|\leq 4 {\sf OPT}$. Towards this observe that 
\begin{eqnarray*}
|S|=|S_1|+|S_2| & \leq & 4|{\cal S}| + 2 |S_{\sf OPT}-S_1|\\
  & \leq & 4 |S_{\sf OPT} \cap S_1|  + 2 |S_{\sf OPT}-S_1| \\
   & \leq & 4 |S_{\sf OPT}|=4{\sf OPT}. \\
\end{eqnarray*}
This completes the proof. 
\qed 
\end{proof}

\section{Improved Kernel for Block Graph Vertex Deletion}
In this section, we give a kernel of $\OO(k^4)$ vertices for  \bgvd. Let $(G,k)$ be an instance of the \bgvd~problem. We start with some of the known reduction rules from~\cite{bgvd-kim}. 

\begin{redrulebgvd}\label{redn:bock-component-rule} 
 If $G$ has a component $H$, where $H$ is a block graph, then remove $H$ from $G$.
\end{redrulebgvd}

\begin{redrulebgvd}\label{redn:cut-vertex-rule}
If there is a vertex $v\in V(G)$, such that $G\setminus \{v\}$ has a component $H$, where $G[\{v\}\cup V(H)]$ is a connected block graph then, remove $H$ from $G$.
\end{redrulebgvd}

\begin{redrulebgvd}\label{redn:twin-rule}
Let $S\subseteq V(G)$, where each $u,v\in S$ are true-twins in $G$. If $\lvert S \rvert > k+1$, then remove all the vertices from $S$ except $k+1$ vertices.
\end{redrulebgvd}

\begin{redrulebgvd}\label{redn:induced-path}
Let $t_1,t_2,t_3,t_4$ be an induced path in $G$. For $i\in \{1,2,3\}$, let $S_i\subseteq V(G)\setminus \{t_1,t_2,t_3,t_4\}$ be a clique in $G$ such that the following holds.
\begin{itemize}
\item For $i\in \{1,2,3\}$, $v\in S_i$, $N_G(v)\setminus S_i=\{t_i,t_{i+1}\}$, and
\item For $i\in \{2,3\}$, $N_G(t_i)=\{t_{i-1},t_{i+1}\} \cup S_{i-1}\cup S_i$.
\end{itemize}
Remove $S_2$ from $G$ and contract the edge $(t_2,t_3)$.
\end{redrulebgvd}

\begin{proposition}[Proposition 3.1 \cite{bgvd-kim}]\label{find-Sv}
Let $G$ be a graph and $k$ be a positive integer. For a vertex $v\in V(G)$, in $\OO(kn^3)$ time, we can find one of the following.
\begin{enumerate}[i.]
\item $k+1$ pairwise vertex disjoint obstructions,
\item $k+1$ obstructions whose pairwise intersection is exactly $v$, 
\item $S'_v\subseteq V(G)$, such that $\lvert S'_v \rvert \leq 7k$ and $G\setminus S'_v$ has no block graph obstruction containing $v$.
\end{enumerate}
\end{proposition}

\begin{redrulebgvd}\label{redn:distinct-obstruction}
Let $v\in V(G)$ and $G'=G \setminus \{v\}$. We remove the edges between $N_G(v)$ from $G'$, i.e. $E(G')=E(G')\setminus \{(u,w)\lvert u,w \in N_G(v)\}$. In $G'$ if there are at least $2k+1$ vertex-disjoint $N_G(v)$-paths in $G'$ then we do one of the following.
\begin{itemize} 
\item If $G$ contains $k+1$ vertex disjoint obstructions, then return that the graph is a no-instance. 
\item Otherwise, delete $v$ from $G$ and decrease $k$ by $1$.
\end{itemize}
\end{redrulebgvd}

The Reduction rules \bgvd.\ref{redn:bock-component-rule} to \bgvd.\ref{redn:distinct-obstruction} are safe and can be applied in polynomial time~\cite{bgvd-kim}. For sake of clarity we denote the reduced instance at each step by $(G,k)$. We always apply the lowest numbered Reduction Rule, in the order that they have been stated, that is applicable at any point of time. For the rest of the discussion, we assume that Reduction rules \bgvd.\ref{redn:bock-component-rule} to \bgvd.\ref{redn:distinct-obstruction} are not applicable.

For a vertex $v\in V(G)$, by Proposition~\ref{find-Sv}, we may find $k+1$ pairwise vertex-disjoint obstructions, and we can safely conclude that the graph is a \no{} instance. Secondly, if we find $k+1$ obstructions whose pairwise intersection is exactly $v$ then the Reduction rule \bgvd.\ref{redn:distinct-obstruction} will be applicable. Thus, we assume that for each vertex $v\in V(G)$, the third condition of Proposition~\ref{find-Sv} holds. In other words, we have a set $S'_v$ of size at most $7k$, such that $G\setminus S'_v$ does not contain any obstruction passing through $v$.  In fact, for each $v\in V(G)$, we can find a \bvds\ $S_v \subseteq V(G) \setminus \{v\}$ of bounded size. 

\begin{observation}\label{obs:approx_soln_S_v}
For every vertex $v\in V(G)$, we can find in $n^{\OO(1)}$ time, a set $S_v \subseteq V(G) \setminus \{v\}$ such that $\lvert S_v \rvert \leq 11k$ and $G\setminus S_v$ is a block graph.
\end{observation}

\begin{proof}
 Using the approximation algorithm for \bgvd~we compute an approximate solution $A$ of size at most $4k$. If $v\notin A$, then $S_v=A$. Otherwise, $S_v=(A \setminus \{v\})\cup S'_v$. Note that for each $v\in V(G)$, $\lvert S_v \rvert \leq 11k$ and $G\setminus S_v$ is a block graph.
\qed
\end{proof}

For a vertex $v \in V(G)$, \compdeg{} of $v$ is the number of connected components in $\CC$, where $\CC$ is the set of connected components in $G\setminus (S_v \cup \{v\})$ that have a vertex adjacent to $v$. We give a reduction rule that bounds the \compdeg{} of a vertex $v \in V(G)$, using \emph{Expansion Lemma}~\cite{saurabh-book}.

A {\em $q$-star}, $q \geq 1$, is a graph with $q + 1$ vertices, one vertex of degree $q$ and all other vertices of degree $1$. Let $\BB$ be a bipartite graph
with the vertex bipartition as $(X,Y)$. A set of edges $M \subseteq E(\BB)$ is called a {\em $q$-expansion} of $X$ into $Y$ if (i) every vertex of $X$ is incident with exactly $q$ edges of $M$ and (ii) $M$ saturates exactly $q|X|$ vertices in $Y$, i.e. edges in $M$ are adjacent to exactly $q|X|$ vertices in $Y$.

\begin{lemma}[Expansion Lemma~\cite{saurabh-book}]\label{qexpansion}
Let $q$ be a positive integer and $\BB$ be a bipartite graph with vertex bipartition $(X,Y)$ such that $\lvert Y \rvert \geq q \lvert X \rvert$ and there are no isolated vertices in $Y$. Then, there exist nonempty vertex sets $X' \subseteq X$ and $Y' \subseteq Y$ such that:
\begin{enumerate}
\item $X'$ has a $q$-expansion into $Y'$ and
\item no vertex in $Y'$ has a neighbour outside $X'$, i.e. $N(Y') \subseteq X'$.
\end{enumerate}
Furthermore, the sets $X'$ and $Y'$ can be found in polynomial time.
\end{lemma}
 
For a vertex $v \in V(G)$, let $\CC_v$ be the set of connected components in $G\setminus (S_v \cup \{v\})$ that have a vertex adjacent to $v$. Consider a connected component $C \in \CC_v$, such that no vertex $u \in V(C)$ is adjacent to any vertex in $S_v$. But then, $G\setminus \{v\}$ has a component which is a \blockgraph{} (namely, the connected component $C$) therefore, Reduction rule \bgvd.\ref{redn:cut-vertex-rule} is applicable, a contradiction to the assumption that none of the previous Reduction rules are applicable. Therefore, for each $C\in \CC$ there is a vertex $u \in V(C)$ and $s \in S_v$, such that $(u,s)\in E(G)$. Let $\DD$ be a vertex set, with a vertex $d$ corresponding to each component $D \in \CC$. Consider the bipartite graph $\BB_v$ with the vertex set bipartitioned as $(\DD,S_v)$. There is an edge between $d\in \DD$ and $s \in S_v$ if and only if the component $D$ corresponding to which the vertex $d$ was added to $\DD$ has a vertex $u_d$ such that $(u_d,s)\in E(G)$.

\begin{redrulebgvd}\label{redn:reduce-no-components}
For a vertex $v \in V(G)$ if $|\CC_v| > 33k$, then we do the following.
\begin{itemize}
\item Let $\DD' \subseteq \DD$ and $S \subseteq S_v$ be the sets obtained after applying Lemma~\ref{qexpansion} with $q=3$, $X=S_v$ and $Y=\DD$;
\item For each $d\in \DD'$, let the component corresponding to $d$ be $D \in \CC_v$. Delete all the edges between $(u,v)$, where $u\in V(D)$;
\item For each $s\in S$, add two vertex disjoint paths between $v$ and $s$.
\end{itemize}
\end{redrulebgvd}

Safeness of the Reduction rule \bgvd.\ref{redn:reduce-no-components} follows from the safeness of Reduction rule 6 in~\cite{bgvd-kim}.

\subsection{Bounding the number of blocks in $G\setminus A$}
Using the approximation algorithm for \bgvd~we compute an approximate solution $A$ of size at most $4k$.  
Of course if $|A|>4k$ then we can immediately return that $G$ is a \no{} instance. 
First, we bound the number of leaf blocks in $G\setminus A$, when none of the Reduction rules apply. Note that $G\setminus A$ is a \blockgraph{}, since $A$ is an approximate solution to  \bgvd. For $v\in A$, let $S'_v$ be the set obtained from Proposition~\ref{find-Sv} and $S_v$ be the set obtained from Observation~\ref{obs:approx_soln_S_v}. Let $\CC_v$ be the set of connected components in $G\setminus (S_v\cup \{v\})$ which have a vertex adjacent to $v$. All the connected components in $G\setminus A$, which do not have a vertex that is adjacent to $v$, must be adjacent to some $v'\in A$. Otherwise, Reduction rule \bgvd.\ref{redn:bock-component-rule} will be applicable.  Also, all the leaf blocks in $G\setminus A$ must have an internal vertex that is adjacent to some vertex in $A$, since the Reduction rules \bgvd.\ref{redn:bock-component-rule} and \bgvd.\ref{redn:cut-vertex-rule} are not applicable. The number of leaf blocks, in $G \setminus A$, whose set of internal vertices have a non-empty intersection with $S'_v$, is at most $7k$. Therefore, it is enough to count, for each $v \in A$, the number of leaf blocks in $\CC_v$. In the Observation~\ref{leaves-equal}, we give a bound on the number of leaf blocks in $G \setminus A$, not containing any vertex from $S'_v$.

\begin{observation}\label{leaves-equal}
For $v \in A$, the number of leaf blocks in $G\setminus A$ not containing any vertex from $S'_v$ is at most the number of leaf blocks in $G\setminus (S_v \cup \{v\})$. 
\end{observation}
\begin{proof}
Note that for $v\in A$, $S_v=(A\setminus \{v\}) \cup S'_v$. By deleting a vertex $u\in S'_v$ from $G\setminus A$ one of the following can happen. If $u$ was a cut vertex in $G \setminus A$, then we increase the number of components after deleting $u$ from $G\setminus A$. By increasing the number of components, the number of leaves cannot decrease. If $u$ was not a cut vertex in $G\setminus A$ then by deleting $u$ the number of cut vertices can only increase. Therefore, the number of leaf blocks after deletion of $u$ from $G\setminus A$ can only increase. Hence, the claim follows.
\qed
\end{proof}

Therefore, for each $v\in A$ we count those leaf blocks in $\CC_v$ which do not contain any vertex from 
$S'_v$. 

\begin{lemma}\label{one-component-neighbour}
Consider a vertex $v\in V(G)$ and its corresponding set $S_v$. Let $\CC$ be the set of connected components in $G\setminus (S_v \cup \{v\})$. For each $C\in \CC$, there is a block $\tilde B$ in $C$, such that $N_C(v)\subseteq V(\tilde B)$. 
\end{lemma}
\begin{proof}
Let $\CC$ be the set of connected components of $G\setminus (S_v \cup \{v\})$, $v\in V(G)$. By definition of $S_v$, for each $C\in \CC$, $C \cup \{v\}$ is a \blockgraph. 

If for some $C\in \CC$, $N_C(v)=\emptyset$, then the condition is trivially satisfied for that connected component $C$. Let $C\in \CC$ be a connected component such that $N_C(v)\neq \emptyset$. Let $t$ be a vertex in $N_B(v)$, where $B$ is a block in $C$. Let $B'$ be a block in $C$, where $B'\neq B$ and $B'$ has a vertex $t' \in V(B') \setminus V(B)$ that is adjacent to $v$. Note that $B,B'$ are in the same connected component $C$. Let $P$ be the shortest path from $t$ to $t'$. 

We first argue for the case when $(t,t')\notin E(G)$. Therefore, the path $P$ has at least $2$ edges. We prove that we can find an obstruction, by induction on the length of the path (number of edges). If length of path $P$ is $2$, say $P=t,u,t'$. If $(u,v)\in E(G)$, then $\{t,t',u,v\}$ forms an induced $D_4$, otherwise they form an induced $C_4$, contradicting that $C \cup \{v\}$ is a \blockgraph.

Let us assume that we can find an obstruction if the path length is $l$. We now prove it for paths of length $l+1$. Let $P=t,x_1,x_2,\dots,x_{l-1},t'$ and $y$ be the first vertex other than $t$ in $P$ such that $(y,v)\in E(G)$. If $y=t'$, then $P$ along with $v$ forms an induced cycle of length at least $5$, contradicting that $C\cup \{v\}$ is a \blockgraph. If $y=x_1$, then $\{t,x_1,x_2,v\}$ either forms a $D_4$, the case when $(x_2,v)\in E(G)$, or $\hat P=x_1,x_2,\dots,t'$ is a path of shorter length with at least $2$ edges and by induction hypothesis has an obstruction along with $v$. Otherwise, $P'=t,x_1,\dots,y$ is a path of length less than $l$, with at least $2$ edges, such that $(y,t)\in E(G)$. Therefore, by induction hypothesis there is an obstruction along with the vertex $v$, contradicting that $C\cup \{v\}$ is a \blockgraph.

From the above arguments it follows that if $v$ has a neighbour $t$ in block $B$ in $C$, then $v$ cannot have a neighbour $t'$ in block $B'$, if the shortest path between $t,t'$ has at least $2$ edges. 

If $(t,t')\in E(G)$, then $t,t'$ are contained in some block $\hat B$. If $v$ is adjacent to any other vertex $u$ not in $V(\hat B)$ then at most one of $(t,u)$ or $(t',u)$ can be an edge in $G$, since $t,t'$ and $u$ are in different blocks. If there is an edge, say $(t,u)$, then $t,t',u,v$ forms an induced $D_4$, contradicting that $C\cup \{v\}$ is a \blockgraph. Otherwise, there is a path with at least two edges between $u$ and $t$. Therefore, by the previous arguments we can find an obstruction along with the vertex $v$. Therefore, $N_C(v)\subseteq V(\hat B)$ when $(t,t')\in E(G)$. 

Hence, it follows that there is a block $\tilde B$ in $C$ such that $N_C(v) \subseteq V(\tilde B)$.
\qed
\end{proof}

\begin{lemma}\label{no-of-leaves}
For every $v\in A$ , the number of leaf blocks in $\CC_v$ is $\OO(k)$.
\end{lemma}
\begin{proof}
 Note that every leaf block must have at least one internal vertex. By Lemma~\ref{one-component-neighbour} we know that neighbours of $v$ are contained in a block of $C$, where $C\in \CC_v$. Therefore, $v$ cannot be adjacent to internal vertices of two leaf blocks in $C\in \CC_v$. In other words, $v$ can be adjacent to vertices in at most one leaf block from $C\in \CC_v$. But $\lvert \CC_v \rvert \leq 33k$, since the Reduction rule \bgvd.\ref{redn:reduce-no-components} is not applicable. Therefore, the number of leaf blocks in $G\setminus (S_v\cup \{v\})$ in which $v$ can have a neighbour is at most $\OO(k)$. 
  \qed
 \end{proof}
 
Observe that in $G\setminus A$, a vertex $v\in A$ can be adjacent to at most $\OO(k)$ leaf blocks by Observation~\ref{leaves-equal} and Lemma~\ref{no-of-leaves}. Also, for a leaf block $B$ in $G\setminus A$, there must be an internal vertex $b\in V(B)$, such that $b$ is adjacent to some vertex in $S_v$, since the Reduction rule \bgvd.\ref{redn:cut-vertex-rule} is not applicable. Therefore, the number of leaf blocks in $G\setminus A$ is $\OO(k^2)$.

\begin{lemma}\label{no-of-degree-3-blocks}
The number of blocks $B$ in $G \setminus A$ such that the vertex set of $B$ intersects with the vertex set of at least three other block in $G\setminus A$ is $\OO(k^2)$.
\end{lemma}
\begin{proof}
Consider the block forest $F_A$ for the \blockgraph{} $G \setminus A$. The number of blocks in $G\setminus A$ is at most the number of vertices in $F_A$. Note that the leaves in $F_A$ correspond to the blocks in $G\setminus A$ with at most one cut vertex. The number of leaf blocks is $G'$ is bounded by $\OO(k^2)$ and therefore the number of leaves in $F_A$ is $\OO(k^2)$. For a forest the number of vertices of degree at least $3$ is bounded by the number of leaves. Therefore, the number of degree three vertices in $F_A$ is bounded by $\OO(k^2)$. For a block $B$ in $G\setminus A$ which has at least $3$ cut vertices, the vertex $b$ corresponding to block $B$ in $F_A$ will be of degree at least $3$. Therefore, the number of blocks in $G\setminus A$ with at least three cut vertices is bounded by $\OO(k^2)$. 

Consider a block $B$ in $G\setminus A$, such that $B$ has exactly $2$ cut vertex, but $V(B)$ intersects with at least three blocks in $G\setminus A$. Let $b$ be the vertex corresponding to $B$ in $F_A$ and $v_{B},u_{B}$ the cut-vertices in $B$. At least one of $v_{B},u_{B}$ is a cut-vertex in at least two blocks other than $B$, say $v_{B}$ is such a cut vertex. If we contract the edge $(b,v_B)$ in $F_A$ to obtain $F'_A$, then number of leaves and degree $3$ vertices in $F'_A$ remains the same which is bounded by $\OO(k^2)$. Furthermore, the contracted vertex $b^*$ is of degree at least $3$. There is a bijection $f$ between the vertices corresponding to blocks in $F_A$ and $F'_A$, where $f(b)=b^*$ and $f(b')=b'$, for all $b'\in V_B(F_A)\setminus \{b\}$, where $V_B(F_A)$ is the set of vertices corresponding to blocks in $G\setminus A$. Hence follow that the number of blocks $B$ with exactly $2$ cut vertices, but $V(B)$ intersects with at least three blocks in $G\setminus A$ is bounded by $\OO(k^2)$.
\qed
\end{proof}

Let $\cal L$ be the set of leaf blocks in  $G\setminus A$ and $\cal T$ be the set of  blocks in $G\setminus A$ such that each block in $\cal T$ intersects with at least three other blocks in $G\setminus A$. By Lemmas~\ref{no-of-leaves} and \ref{no-of-degree-3-blocks}, we have that $|{\cal L}|= \OO(k^2)$ and $|{\cal T}|= \OO(k^2)$. 

%

Let $B$ be a block in $G^*=G\setminus (S_v\cup \{v\})$ such that the vertex set of $B$ has exactly two cut vertices, and intersects with exactly two blocks of $G^*$. Furthermore, the vertex set of $B$ has an empty intersection with leaf blocks of $G^*$ and those blocks in $G^*$ which vertex set intersects with at least three other blocks of $G^*$. Also, $B$ has a vertex that is a neighbor of $v$. Such blocks are called {\em nice degree two blocks} of $v$. If a block satisfies the above conditions for some vertex $w \in A$, the block is called a {\em nice degree two block}. We denote the set of nice degree two blocks by ${\cal{T}}_1$.

\begin{lemma} \label{deg-2-blocks}
Let $G^*=G\setminus (S_v\cup \{v\})$. Then $G^*$ has at most $\OO(k)$ nice degree two blocks of $v$. 
%
\end{lemma}
\begin{proof}
Recall that  $\CC_v$ is the set of connected components in $G\setminus (S_v\cup \{v\})$ which have a vertex adjacent to $v$.  
From Lemma~\ref{one-component-neighbour}, for each of the connected component $C\in \CC_v$, there is a block $\tilde B$ in $C$ such that, $N_C(v)\subseteq V(\tilde B)$.  Consider two nice degree two blocks $B$ and $B'$ in $C$ such that both of them have at least one neighbor of $v$. 
Let $b\in V(B)$ and $b'\in V(B')$ be the vertices such that $(v,b),(v,b')\in E(G)$. Consider the following cases.
\begin{itemize}
\item $b\neq b'$. In this case, $b,b' \in V(\tilde B)$, where $N_C(v)\subseteq V(\tilde B)$ and $\tilde B$ is a block with exactly $2$ cut vertices. But then, for all blocks $\hat B$ in $C$, $V(\hat B)\cap V(\tilde B)=\emptyset$. Therefore, $v$ cannot be adjacent to any vertex in $\hat B$. Hence it follows that the number of blocks with a vertex having neighbour in $v$ is $\OO(1)$.
\item If $b=b'$, then $b$ is a cut vertex in $B,B'$. Therefore, one of $B,B'$ is same as $\tilde B$, say $B=\tilde B$. But $B$ shares cut vertices with exactly two blocks. Therefore, $v$ can have neighbous in at most $\OO(1)$ blocks in $C$. 
\end{itemize}
But note that $|\CC_v|=\OO(k)$. Hence,  the claim follows.
\qed
\end{proof}

What remains is to bound the number of blocks which have exactly two cut vertices and are not nice degree two blocks. 

\begin{lemma}\label{no-of-degree-2-blocks}
The number of blocks in $G \setminus A$ with exactly two cut vertices is $\OO(k^2)$.
\end{lemma}
\begin{proof}
From Lemma~\ref{no-of-leaves} the number of leaf blocks in $G \setminus A$ is $\OO(k^2)$. Let $F_A$ be the block forest for the \blockgraph{} $G \setminus A$. From Lemmas~\ref{no-of-leaves} and \ref{no-of-degree-3-blocks}, we know that $|{\cal L} \cup {\cal T}|= \OO(k^2)$. Also, the number of blocks in ${\cal{T}}_1$ is bounded by $\OO(k^2)$ by Lemma~\ref{deg-2-blocks}. 

Let $\mathcal{P}$ be the set of paths in $F_A$ such that the endpoints are vertices corresponding to blocks in ${\cal L} \cup {\cal T} \cup {\cal{T}}_1$ and all internal block vertices do not correspond to blocks in ${\cal L} \cup {\cal T} \cup {\cal{T}}_1$. Note that, all internal vertices of such paths have degree exactly two in $F_A$. Since $F_A$ is a tree, the number of paths in $\mathcal{P}$ is at most $\OO(k^2)$. 


 Denote the remaining blocks with exactly two cut vertices by ${\cal{T}}_2$. By definition of $F_A$ and $\mathcal{P}$, the vertex corresponding to a block $B \in {\cal{T}}_2$ must be an internal vertex in some path of $\mathcal{P}$. Let $P_A$ be a path in $\mathcal{P}$, Note that $F_A$ is a bipartite graph with the vertex bipartitions as $(\BB,V_C)$, where $\BB$ is the set of blocks in $G'$ and $V_C$ is the set of cut vertices in $G'$. Therefore, in $P$ no two cut vertices in $V(G')$ can be adjacent to each other. Similarly for $b,b'\in \BB$, $(b,b')\notin E(P)$. Therefore, the $b \in V(P)$ such that $b\in \BB$ corresponds to block $B$ in $G'$ with exactly two cut vertices. Let $\mathcal{S}$ be the sequence of blocks in the order they appear in path $P_A$. In $\mathcal{S}$ each two adjacent blocks in the sequence share a cut vertex. We remove the starting block and the end block from $\mathcal{S}$. We do this because these blocks in the subsequence either belong to ${\cal L} \cup {\cal T} \cup {\cal{T}}_1$. If $\mathcal{S}$ has more than two blocks then, Reduction rule \bgvd.\ref{redn:induced-path} would be 
applicable. Therefore, the number of blocks in $\mathcal{P}$ can be at most $\OO(1)$.

The set of blocks in $G \setminus A$ with exactly two cut vertices is contained in ${\cal {T}} \cup {\cal{T}}_1 \cup {\cal{T}}_2$. Hence, it follows that the number of blocks with exactly $2$ cut vertices in $G\setminus A$ is bounded by $\OO(k^2)$.
\qed
\end{proof}

Now, we have a bound on the total number of blocks in $G \setminus A$.
\begin{lemma}\label{bounding-no-of-blocks}
Consider a graph $G$, a positive integer $k$ and an approximate \bvds\ set $A$ of size $\OO(k)$. If none of the Reduction rules \bgvd.\ref{redn:bock-component-rule} to \bgvd.\ref{redn:reduce-no-components} is applicable then the number of blocks in $G\setminus A$ is bounded by $\OO(k^2)$.
\end{lemma}
\begin{proof}
Follows from Lemmas~\ref{no-of-leaves},~\ref{no-of-degree-3-blocks} and \ref{no-of-degree-2-blocks}.
\qed
\end{proof}

\subsection{Bounding the number of internal vertices in a maximal clique of the \blockgraph}
We start by bounding the number of internal vertices in a maximal $2$-connected component of $G\setminus A$. Consider a block $B$ in $G\setminus A$. We partition the internal vertices $V_I(B)$ of block $B$ into three sets $\BB,\RR$ and $\II$ depending on the neighborhood of $A$ in block $B$. We also partition the vertices in $A$ depending on the number of vertices they are adjacent to in $B$. In Lemma~\ref{bounding-no-of-vertices} we show that the number of internal vertices in a block $B$ of 
$G\setminus A$ is upper bounded by $\OO(k^2)$. We do so by partitioning the vertices into different sets and bounding each of these sets separately.

\begin{lemma}\label{bounding-no-of-vertices}
Let $(G,k)$ be an instance to \bgvd\ and let $A$ be an approximate \bvds\ of $G$ of size $\OO(k)$. 
If none of the Reduction rules \bgvd.\ref{redn:bock-component-rule} to \bgvd.\ref{redn:reduce-no-components} is applicable then the number of internal vertices in a block $B$ of $G\setminus A$ is bounded by $\OO(k^2)$.
\end{lemma}
\begin{proof}
Let $A_{\leq 2k}=\{v \in A | \lvert N_B(v)\rvert \leq 2k+1\}$ and $A_{>2k+1}=A \setminus A_{\leq 2k+1}$. For a vertex $u \in V_I(B)$, if $u$ is adjacent to at least one of the vertices in $A_{\leq 2k+1}$ then, we add $u$ to the set $\BB$. Note that the number of vertices in $\BB$ is bounded by $\OO(k^2)$, since each $v\in A_{\leq 2k+1}$ is adjacent to at most $2k+1$ vertices in $V_I(B)$. Also, for a vertex $u\in V_I(B)\setminus \BB$, $N(u)\cap A_{\leq 2k}=\emptyset$. 

For each vertex $u\in V_I(B)\setminus \BB$, if $u$ is not adjacent to at least one vertex in $A_{>2k+1}$ then, we add $u$ to the set $\RR$. Let $Q_v$ be those vertices in $V_I(B)\setminus \BB$ which are not adjacent to a vertex $v\in A_{>2k+1}$. Note that $\lvert Q_v \rvert \leq k$ otherwise, for each pair of vertices $t_1,t_2 \in N_B(v)$ along with one vertex in $Q_v$ we get $k+1$ vertex disjoint obstruction, namely $D_4$, intersecting only at $v$. Therefore, the number of vertices in $\RR$ is bounded by $\OO(k^2)$.

Let $\II=V_I(B)\setminus (\BB \cup \RR)$. Note that the vertices in $\cal I$ induce a clique. Furthermore, for each $w \in \II$, $N(w) \cap A_{\leq 2k+1}=\emptyset$ and $N(w)\cap A_{>2k+1}=A_{>2k+1}$. Therefore, each $w,w'\in \II$ are twins. In fact,  $\cal I$ is a set of twins. If $\vert \II \rvert >k+1$ then Reduction rule \bgvd.\ref{redn:twin-rule} would be applicable. Therefore, $|\II|\leq k+1$. But, $|V_I(B)|= |\BB|+|\RR|+|\II|=\OO(k^2)+\OO(k^2)+\OO(k)$, which is bounded by $\OO(k^2)$.
\qed
\end{proof}

We wrap up our arguments to show a $\OO(k^4)$ sized vertex kernel for \bgvd, and hence prove Theorem~\ref{thm:bgvd-kernel}.

\begin{proof}[of Theorem~\ref{thm:bgvd-kernel}]
Let $(G,k)$ be an instance to \bgvd\ and let $A$ be an approximate \bvds\ of $G$ of size $\OO(k)$. 
Also, assume that none of the Reduction rules \bgvd.\ref{redn:bock-component-rule} to \bgvd.\ref{redn:reduce-no-components} are applicable. By Theorem~\ref{bounding-no-of-blocks}, the number of blocks in $G\setminus A$ is bounded by $\OO(k^2)$. By Lemma~\ref{bounding-no-of-vertices} the number of internal vertices in a block of $G\setminus A$ is bounded by $\OO(k^2)$. Also note that the number of cut-vertices in $G \setminus A$ is bounded by the number of blocks in $G\setminus A$, i.e. $\OO(k^2)$. The number of vertices in $G\setminus A$ is sum of the internal vertices in $G\setminus A$ and the number of cut vertices in $G\setminus A$. Therefore,
$|V(G)|=|V(G\setminus A)|+|A|= (\OO(k^2)\cdot \OO(k^2)+\OO(k^2))+\OO(k)=\OO(k^4)$.
\qed
\end{proof}

\bibliographystyle{abbrv}
\bibliography{reference,references}

\end{document}